\documentclass[aps,prl,twocolumn,superscriptaddress,floatfix,nofootinbib,showpacs,longbibliography,groupedaddress]{revtex4-2}

\usepackage[utf8]{inputenc}  
\usepackage[T1]{fontenc}     
\usepackage[british]{babel}  
\usepackage[sc,osf]{mathpazo}
\usepackage[table]{xcolor}
\usepackage[scaled=0.86]{berasans}  
\usepackage[colorlinks=true, citecolor=blue, urlcolor=blue]{hyperref}
\makeatletter
\newcommand{\setword}[2]{%
  \phantomsection
  #1\def\@currentlabel{\unexpanded{#1}}\label{#2}%
}
\makeatother
\usepackage{graphicx} 
\usepackage[babel]{microtype}
\usepackage{blkarray}
\usepackage{amsmath}

\usepackage{amsmath,amssymb,amsthm,bm,amsfonts,mathrsfs,bbm} 

\usepackage{xspace}  
\usepackage{pgf,tikz}
\usepackage{xcolor}
\usepackage{multirow}
\usepackage{array}
\usepackage{bigstrut}
\usepackage{braket}
\usepackage{color}
\usepackage{natbib}

\usepackage{multirow}
\usepackage{mathtools}
\usepackage{float}
\usepackage[caption = false]{subfig}
\usepackage{xcolor,colortbl}
\usepackage{color}

\newcommand{\be}{\begin{equation}}
\newcommand{\ee}{\end{equation}}
\newcommand{\ba}{\begin{eqnarray}}
\newcommand{\ea}{\end{eqnarray}}

\newtheorem{theorem}{Theorem}

\newtheorem{proposition}{Proposition}

\usepackage{diagbox}
\usepackage{multirow}
\usepackage{tabularx}
\usepackage{comment}

\usepackage{slashbox}
\usepackage{tikz}
\usetikzlibrary{decorations.markings}
\usepackage{pgfplots}
\usepackage{verbatim}
\usepgfplotslibrary{fillbetween}
\usetikzlibrary{patterns.meta}
\usetikzlibrary{patterns}

\usepackage{blindtext}
\usepackage{tcolorbox}

\def\MarkLt{4pt}
\def\MarkSep{2pt}

\tikzset{
  TwoMarks/.style={
    postaction={decorate,
      decoration={
        markings,
        mark=at position #1 with
          {
              \begin{scope}[xslant=0.2]
              \draw[line width=\MarkSep,white,-] (0pt,-\MarkLt) -- (0pt,\MarkLt) ;
              \draw[-] (-0.5*\MarkSep,-\MarkLt) -- (-0.5*\MarkSep,\MarkLt) ;
              \draw[-] (0.5*\MarkSep,-\MarkLt) -- (0.5*\MarkSep,\MarkLt) ;
              \end{scope}
          }
       }
    }
  },
  TwoMarks/.default={0.5},
}

\def\>{\rangle}
\def\<{\langle}

\usepackage{stmaryrd}

\usepackage{centernot}
\usepackage{subfig}

\usepackage{diagbox}
\usepackage{multirow}
\usepackage{tabularx}

\begin{document}  

\title{Scalable \& Noise-Robust Communication Advantage of Multipartite Quantum Entanglement}

\author{Ananya Chakraborty}
\email{ananyaphys.c@gmail.com}
\affiliation{Department of Physics of Complex Systems, S. N. Bose National Center for Basic Sciences, Block JD, Sector III, Salt Lake, Kolkata 700106, India.}

\author{Ram Krishna Patra}
\affiliation{Department of Physics of Complex Systems, S. N. Bose National Center for Basic Sciences, Block JD, Sector III, Salt Lake, Kolkata 700106, India.}

\author{Kunika Agarwal}
\affiliation{Department of Physics of Complex Systems, S. N. Bose National Center for Basic Sciences, Block JD, Sector III, Salt Lake, Kolkata 700106, India.}

\author{Samrat Sen}
\affiliation{Department of Physics of Complex Systems, S. N. Bose National Center for Basic Sciences, Block JD, Sector III, Salt Lake, Kolkata 700106, India.}

\author{Pratik Ghosal}
\affiliation{Department of Physics of Complex Systems, S. N. Bose National Center for Basic Sciences, Block JD, Sector III, Salt Lake, Kolkata 700106, India.}

\author{Sahil Gopalkrishna Naik}
\affiliation{Department of Physics of Complex Systems, S. N. Bose National Center for Basic Sciences, Block JD, Sector III, Salt Lake, Kolkata 700106, India.}

\author{Manik Banik}
\affiliation{Department of Physics of Complex Systems, S. N. Bose National Center for Basic Sciences, Block JD, Sector III, Salt Lake, Kolkata 700106, India.}

\begin{abstract}
Distributed computing, involving multiple servers collaborating on designated computations, faces a critical challenge in optimizing inter-server communication—an issue central to the study of communication complexity. Quantum resources offer advantages over classical methods in addressing this challenge. In this work, we investigate a distributed computing scenario with multiple senders and a single receiver, establishing a scalable advantage of multipartite quantum entanglement in mitigating communication complexity. Specifically, we demonstrate that when the receiver and the senders share a multi-qubit Greenberger-Horne-Zeilinger (GHZ) state—a quintessential form of genuine multipartite entanglement—certain global functions of the distributed inputs can be computed with only one bit of classical communication from each sender. In contrast, without entanglement, two bits of communication are required from all but one sender. Consequently, quantum entanglement reduces communication overhead by \(n-1\) bits for \(n\) senders, allowing for arbitrary scaling with an increasing number of senders. We also show that the entanglement-based protocol exhibits significant robustness under white noise, thereby establishing the potential for experimental realization of this novel quantum advantage.
\end{abstract}


\maketitle	
\section{Introduction}
Processing information with quantum systems offers significant advantages over traditional classical methods \cite{Dowling2003,Nielsen2012,Deutsch2020,Aspect2023}. In particular, quantum entanglement -- one of the most intriguing nonclassical features of multipartite quantum systems \cite{Einstein1935,Bohr1935,Schrodinger1935,Schrodinger1936,Bell1964} -- has profound implications for communication \cite{Wilde2013} and distributed computing \cite{Linden2007}. For instance, Holevo's theorem limits the classical capacity of a perfect quantum channel to that of its classical counterpart \cite{Holevo1973} (see also \cite{Frenkel2015,DallArno2017,Naik2022,Patra2023}). However, as demonstrated by the seminal superdense coding protocol \cite{Bennett1992} and subsequent studies \cite{Adami1997,Bennett1999,Bennett2002}, pre-existing entanglement between sender and receiver can enhance the classical capacity of quantum channels.

Pre-shared quantum entanglement can significantly reduce the need for classical communication in communication complexity problems \cite{Ambainis1996,Cleve1997,Buhrman1998,Cleve1999,Buhrman1999,Buhrman2001,Wolf2001}. In these problems, classical inputs are distributed across distant servers (or parties), and a global function of these inputs must be computed at a designated server. Nonlocal correlations \cite{Bell1966,Mermin1993,Brunner2014} obtained from entangled quantum states play crucial role in minimizing classical communication requirements during function computation \cite{Buhrman2010}. Notably, stronger than quantum nonlocal correlations can trivialize communication complexity problems, underscoring a clear distinction between quantum and post-quantum correlations \cite{vanDam2012,Brassard2006,Naik2023,Botteron2024,Sidhardh2024,Ghosh2024}.

In this work, we investigate how far the advantages of quantum entanglement can be extended in communication complexity tasks. Specifically, we examine the utility of multipartite quantum entanglement, with a particular focus on genuine multipartite entanglement \cite{Dur2000, Verstraete2002, Puliyil2022, Joshi2024}. We consider a communication complexity task, $\mathrm{CC}_n$, involving $n$ distant senders, denoted as $\{\text{Alice-}i\}_{i=1}^n$, and a computing server, Bob. The senders and the server are each provided with two-bit strings that satisfy a promise condition, and Bob's task is to evaluate a specific global function of these input strings. We demonstrate that, when the senders and Bob share an \((n+1)\)-qubit Greenberger–Horne–Zeilinger (GHZ) state, Bob can perfectly compute the function with only one bit of classical communication from each sender. In contrast, without entanglement, the $\mathrm{CC}_2$ problem requires two bits of communication from one sender and one bit from the other. Thus, preshared entanglement provides a one-bit advantage in the $\mathrm{CC}_2$ task. For the general $\mathrm{CC}_n$ task, without entanglement, two bits of communication are required from all but one sender. Therefore, preshared entanglement offers a reduction of $(n-1)$-bit communication overhead, which can be made arbitrarily large as the number of parties increases. Additionally, we analyze a noise-robust version of the task to assess the effectiveness of a noisy GHZ state in communication complexity. As it turns out GHZ state under white noise, exhibits the advantage for quite a large noise parameter.

\section{PRELIMINARIES}   
{\it Communication complexity.--} Classical communication complexity addresses the question of minimizing the amount of classical communication required to compute a function with distributed inputs  \cite{Yao1979, Kushilevitz1996, Kushilevitz1997}. In its simplest form, two servers, commonly referred to as Alice and Bob, are provided randomly chosen bit strings, ${\bf x}\in\{0,1\}^{\times m}$ and ${\bf y}\in\{0,1\}^{\times m^\prime}$, respectively. Bob is tasked with computing a function $f:{\bf x}\times{\bf y}\mapsto\{0,1\}$ while minimizing the classical communication required from Alice. To achieve this goal  Alice and Bob are typically allowed to use pre-shared classical correlations, also called shared randomness\cite{Bavarian2014, Canonne2015}. Quantum communication complexity extends this framework by utilizing quantum resources. In the model introduced by Yao \cite{Yao1993}, Alice sends qubits to Bob, while in the entanglement-based model proposed by Cleve and Buhrman \cite{Cleve1997}, the parties share entanglement and communicate using classical bits. The quantum advantage in the entanglement-based model utilizes the nonlocal feature of correlations obtained from the quantum entangled states \cite{Buhrman2010}. In this work, we aim to explore the nontrivial advantages of genuine multipartite quantum entanglement in communication complexity by considering a multipartite variant of the task. Before presenting our main results, we briefly recall the notion of multipartite quantum entanglement.

{\it Multipartite entanglement.--} State of a quantum system \( S \) associated with a complex Hilbert space \( \mathcal{H}_S \) is described by a density operator \(\rho_S\), a trace-one positive semidefinite matrix, within the set \(\mathcal{D}(\mathcal{H}_S)\) of all such operators. For finite dimensions, \(\mathcal{H}_S\) is isomorphic to a complex Euclidean space \(\mathbb{C}^{d_S}\), where \( d_S \) represents the dimension of the Hilbert space. States with \(\mathrm{Tr}[\rho_S^2] = 1\) are called pure states, corresponding to normalized ray vectors \(\ket{\psi}_S \in \mathbb{C}^{d_S}\). Multipartite quantum systems are associated with the tensor product of their subsystems' Hilbert spaces. A state \(\ket{\psi}_{A_1 \cdots A_K} \in \otimes_{j=1}^K \mathbb{C}^{d_j}\) is fully product if it can be expressed as \(\ket{\psi}_{A_1 \cdots A_K} = \otimes_{j=1}^K \ket{\phi}_{A_j}\), where \(\ket{\phi}_{A_j} \in \mathbb{C}^{d_j}\); otherwise, \(\ket{\psi}_{A_1 \cdots A_K}\) contains some form of multipartite entanglement \cite{Horodecki2009, Guhne2009}. For \(K > 2\), various forms of entanglement are possible, with genuine entanglement being the most intricate one \cite{Dur2000, Verstraete2002, Puliyil2022, Joshi2024}. A state \(\ket{\psi}_{A_1 \cdots A_K}\) is called genuinely entangled if it cannot be decomposed into a product form across any bipartition, i.e., \(\ket{\psi}_{A_1 \cdots A_K} \neq \ket{\phi}_{P} \otimes \ket{\chi}_{P^C}\) for all non-empty subsets \(P \subsetneq \{A_1, \cdots, A_K\}\) where \(P^C := \{A_1, \cdots, A_K\} \setminus P\). A canonical example of a genuinely entangled \(K\)-qubit state, that will be used in this work, is the GHZ state \(\ket{G_K} := (\ket{0}^{\otimes K} + \ket{1}^{\otimes K})/\sqrt{2} \in (\mathbb{C}^2)^{\otimes K}\), originally introduced while studying multipartite variant of Bell nonlocality \cite{Greenberger1989} (see also \cite{Mermin1990, Greenberger1990, Pan2000}).

\section{RESULTS}
We start by defining the communication complexity task $\mathrm{CC}_n$, that involves $n$ distant senders $\{\mbox{Alice-}i\}_{i=1}^n$ and a computing server Bob. Each party receives bit strings ${\bf x}_1=x^0_1x^1_1$, ${\bf x}_2=x^0_2x^1_2$, ..., ${\bf x}_{n}=x^0_{n}x^1_{n}$, and ${\bf y}=y^0y^1$, respectively. The bits $x^1_1,\cdots,x^1_{n}$, and $y^1$ are sampled independently and uniformly at random, while the strings otherwise satisfy a constraint known to be the promise condition, that is $(\sum_{i=1}^{n}x^0_i+y^0)$ must be an even number. Bob's goal is to compute the function $f_n:{\bf y}\times_{i=1}^{n}{\bf x}_i\mapsto\{0,1\}$, defined as $f_n({\bf x}_1,\cdots,{\bf x}_{n},{\bf y}):=\oplus_{i=1}^{n}x^1_i\oplus y^1\oplus \mathbb{P}[(\sum_{i=1}^{n}x^0_i+y^0)/2]$, where `$\oplus$' denotes the binary XOR operation, and $\mathbb{P}:\mathbb{N}\mapsto\{0,1\}$ is defined as $\mathbb{P}(r)=1$ when $r$ is odd and $\mathbb{P}(r)=0$ otherwise. The collaborative goal of the parties is to compute the function $f_n$, with minimum possible collaboration. For this we now discuss an efficient protocol utilizing multi-qubit GHZ state apriori shared among the parties.
\begin{theorem}\label{theo1}
The function $f_n$ can be computed exactly by Bob with $1$-bit of classical communication from each of the Alices when the GHZ state $\ket{G_{n+1}}$ is shared among them.    
\end{theorem}
\begin{proof}
The protocol proceeds as follows: Each party performs Pauli-\(X\) (\(\sigma_1\)) or Pauli-\(Y\) (\(\sigma_2\)) measurement on their respective shares of \(\ket{G_{n+1}}\) state, depending on whether the first bit of their respective strings is \(0\) or \(1\). Each of the Alices then communicates XOR of their measurement outcome and the second bit of their respective input strings to Bob. Bob’s final output of the desired computation is XOR of all the communications received from Alices, his second bit, and outcome of his Pauli measurement.

To demonstrate that this protocol correctly evaluates the function \(f_n\), we begin by recalling an interesting property (say \(\mathcal{P}\)) of the GHZ state \cite{Mermin1990(1)}: 
\begin{align}
\sigma_{1+\vec{z}}\ket{G_{n+1}}=\begin{cases}
(-1)^{\frac{S_z}{2}}\ket{G_{n+1}},~~~~~\mbox{for~even}~S_z;
\\
(-1)^{\frac{S_z+1}{2}}i\ket{G^-_{n+1}},~~\mbox{for~odd}~S_z,
\end{cases}
\end{align}
where \(\sigma_{1+\vec{z}} := \otimes_{i=1}^{n}\sigma_{1+x^0_i}\otimes\sigma_{1+y^0}\), \(S_z := \sum_{i=1}^{n}x^0_i + y^0\), and $\ket{G^-_{K}}:=(\ket{0}^{\otimes K}-\ket{1}^{\otimes K})/\sqrt{2}$. Notably, in $\mathrm{CC}_n$ game, the promise condition ensures $S_z$ to be even always. Let \(\mathcal{O}(x^0_i) \in \{0,1\}\) denotes outcome of \(i^{\text{th}}\) Alice's local measurement \(\sigma_{1+x^0_i}\), and similarly let \(\mathcal{O}(y^0) \in \{0,1\}\) denotes Bob's outcome of the local measurement \(\sigma_{1+y^0}\). According to the protocol, \(i^{\text{th}}\) Alice communicates \(c_i = \mathcal{O}(x^0_i) \oplus x^1_i\) to Bob. Bob's final computation thus becomes, $\oplus_{i=1}^{n}c_i\oplus y^1\oplus \mathcal{O}(y^0)= \oplus_{i=1}^{n}[\mathcal{O}(x^0_i)\oplus x^1_i]\oplus y^1\oplus \mathcal{O}(y^0)=[\oplus_{i=1}^{n} x^1_i\oplus y^1]\oplus[\oplus_{i=1}^{n}\mathcal{O}(x^0_i)\oplus \mathcal{O}(y^0)]=\oplus_{i=1}^{n} x^1_i\oplus y^1\oplus\mathbb{P}[S_z/2]=f_n({\bf x}_1,\cdots,{\bf x}_{n},{\bf y})$. Here, the step $\oplus_{i=1}^{n}\mathcal{O}(x^0_i)\oplus \mathcal{O}(y^0)=\mathbb{P}[S_z/2]$ follows from property $\mathcal{P}$. This completes the proof.
\end{proof}
To establish advantage of GHZ in computing the function $f_n$ in $\mathrm{CC}_n$ task we now analyze its computability without the resource $\ket{G_{n+1}}$. We start by analyzing the case of $\mathrm{CC}_2$. 
\begin{theorem}\label{theo2}
With one bit of classical communication from each of the Alices, Bob cannot compute the function $f_2$ exactly, even in assistance with classical shared randomness.
\end{theorem}
\begin{proof}
For the $\mathrm{CC}_2$ task, the promise condition simplifies to $x^0_1 \oplus x^0_2 \oplus y^0 = 0$, and the function Bob needs to evaluate becomes $f_2({\bf x}_1, {\bf x}_2, {\bf y}) = x^1_1 \oplus x^1_2 \oplus y^1 \oplus (x^0_1 \lor x^0_2 \lor y^0)$. Clearly, Bob can compute $f_2$ if and only if he can compute $f^\prime_2({\bf x}_1, {\bf x}_2, y^0) = x^1_1 \oplus x^1_2 \oplus (x^0_1 \lor x^0_2 \lor y^0)$. Given the promise condition, $f^\prime_2$ can be rewritten as:
\begin{subequations}\label{f3}
\begin{align}
f^\prime_2({\bf x}_1, {\bf x}_2, 0) &= x^1_1 \oplus x^1_2 \oplus (x^0_1 \lor x^0_2), ~ \text{with} ~ x^0_1 = x^0_2; \label{f30}\\
f^\prime_2({\bf x}_1, {\bf x}_2, 1) &= x^1_1 \oplus x^1_2 \oplus 1, \quad \quad \quad \text{with} \quad x^0_1 \neq x^0_2. \label{f31}
\end{align}   
\end{subequations}
A general function $g:\{0,1\}^{\times 2} \mapsto \{0,1\}$ can be expressed as $g^{\alpha,\beta,\gamma,\delta}(u, v) := \alpha u \oplus \beta v \oplus \gamma uv \oplus \delta$, with $\alpha, \beta, \gamma, \delta, u, v \in \{0,1\}$. Consequently, the communication from the $i^{th}$ Alice to Bob can be represented as $c_i := \mathcal{E}_i(x^0_i, x^1_i) = \alpha_i x^0_i \oplus \beta_i x^1_i \oplus \gamma_i x^0_i x^1_i \oplus \delta_i$. To compute the desired function in Eq. (\ref{f3}), Bob applies a generic decoding function to the communication bits $c_1$ and $c_2$ received from Alice-$1$ and Alice-$2$, respectively. Notably, Bob's decoding can depend on the input $y^0$. Denoting the respective decoding functions as $\mathcal{D}_j(c_1, c_2) = \alpha^j_3 c_1 \oplus \beta^j_3 c_2 \oplus \gamma^j_3 c_1 c_2 \oplus \delta^j_3$ for $j\in\{0,1\}$, exact computability of $f_2$ thus demands 
\begin{align}
\mathcal{D}_0=f^\prime_2({\bf x}_1,{\bf x}_2,0)~~~~\&~~~~\mathcal{D}_1=f^\prime_2({\bf x}_1,{\bf x}_2,1). \end{align}
A lengthy but straightforward calculation results yields:
\begin{itemize}
\item[(i)] For $y^0=1$, we have the restrictions $\beta_1=\beta_2=\beta^1_3=\alpha^1_3=1,~\gamma_1= \gamma_2=\gamma^1_3=0$ and $\alpha_1=\alpha_2$;
\item[(ii)] For $y^0 = 0$, we have the restrictions $\alpha^0_3 = \beta^0_3 = 1, \gamma^0_3 = 0, \delta_1= \delta_2\oplus\delta^0_3$ and $\alpha_1 \neq \alpha_2$.
\end{itemize}
Since the restrictions in (i) and (ii) are not consistent with each other, therefore no consistent encodings $\mathcal{E}_1,\mathcal{E}_2$ and decodings $\mathcal{D}_0,\mathcal{D}_1$ exist that can evaluate the function $f_2$ exactly. Having no deterministic strategy their probabilistic mixtures also fail to compute the function $f_2$. This completes the proof.   
\end{proof}
Evidently, the function $f_2$ can be evaluated with classical resources if each Alice communicates $2$ bits to Bob. However, as demonstrated in our subsequent result, a more efficient classical protocol is available.
\begin{proposition}\label{prop1}
The function \( f_2 \) can be evaluated exactly with $2$ bits of classical communication from one Alice and $1$ bit from the other.
\end{proposition}
\begin{proof}
With $2$ bits let Alice-$1$ classically communicates her string ${\bf x}_1$ to Bob, while Alice-$2$, using $1$ bit channel, communicates her second bit $x^1_2$ only. The promise condition, $x^0_1\oplus x^0_2\oplus y^0=0$, allows Bob to evaluate $x^0_2$ from the information of $y^0$ and $x^0_1$, and therefore the function $f_2$ can be computed exactly.    
\end{proof}
Proposition \ref{prop1}, along with Theorem \ref{theo1} and \ref{theo2}, demonstrates that the multipartite entangled state \( \ket{G_3} \) provides a $1$-bit communication advantage in evaluating the function \( f_2 \) in the \( \mathrm{CC}_2 \) task. We now proceed to establish a scalable communication advantage of multipartite entanglement by considering the \( \mathrm{CC}_n \) task.
\begin{figure}[t!]
\centering
\includegraphics[width=1\linewidth]{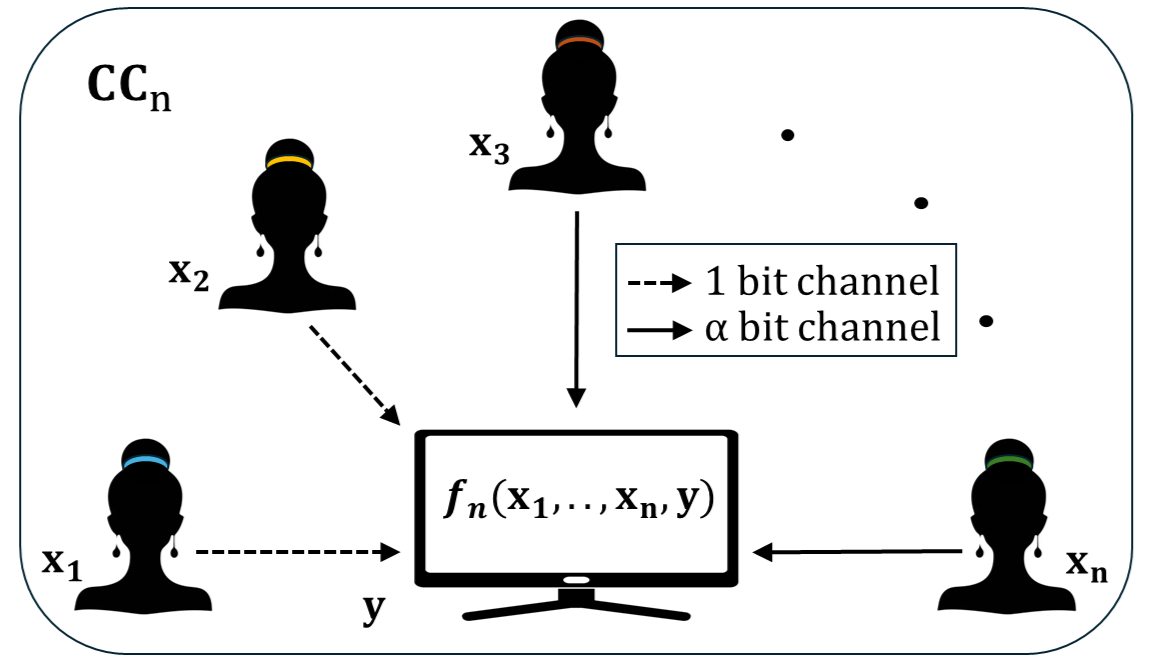}
\includegraphics[width=1\linewidth]{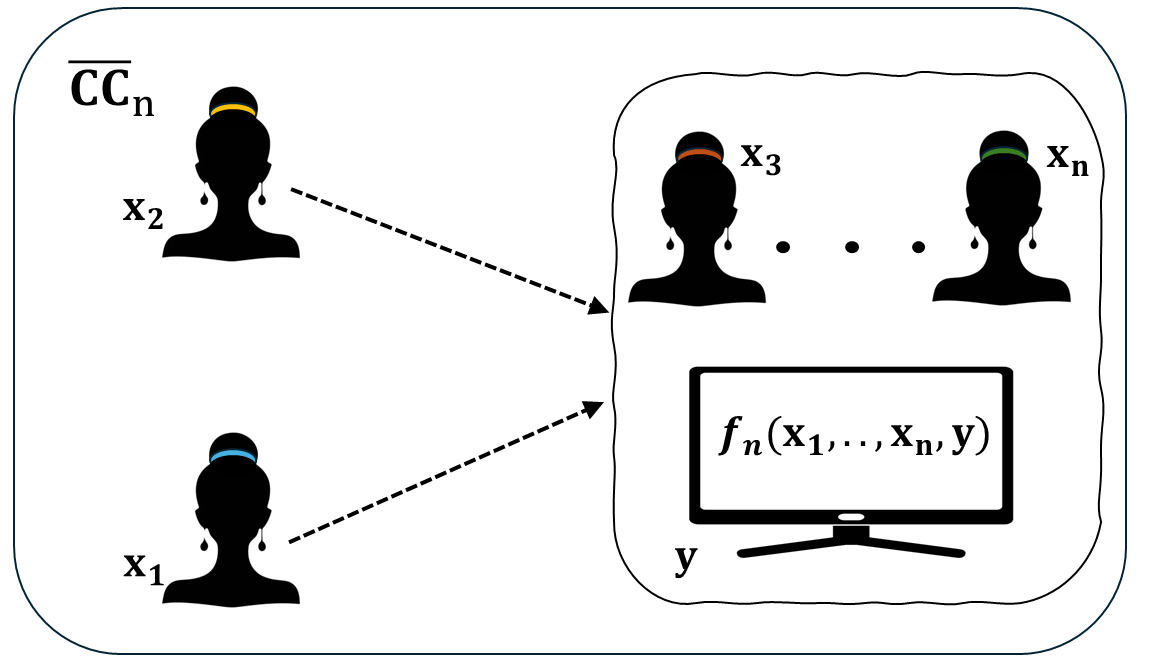}
\caption{Impossibility of evaluating $f_n$ in $\overline{\mathrm{CC}}_n$ scenario implies impossibility of evaluating $f_n$ in $\mathrm{CC}_n$ scenario, for any $\alpha$ channels.}\label{fig}
\vspace{-.5cm}
\end{figure} 
Using a similar argument as in Proposition \ref{prop1}, it is evident that the function \( f_n \) can be evaluated exactly with $1$ bit of communication from one Alice and $2$ bits from each of the others, requiring a total of $( 2n - 1 )$ bits of communication. Referring to Theorem \ref{theo1}, we can thus conclude that the multipartite entangled state \( \ket{G_{n+1}} \) reduces the communication overhead by $( n - 1 )$ bits when evaluating the function \( f_n \). In other words, the communication advantage of multipartite entanglement can be made arbitrarily large by involving more number of senders in the communication complexity task $\mathrm{CC_n}$.
\begin{theorem}\label{theo3}
For $n\ge2$, the function $f_n$ cannot be evaluated by Bob exactly, if two or more Alices are limited to share only $1$-bit of classical channel with Bob.   
\end{theorem}
\begin{proof}
To prove the theorem, it suffices to consider the case of two Alices (Alice-$1$ and Alice-$2$), each sharing a $1$-bit classical channel with Bob. Motivated by the task \(\mathrm{CC}_n\), where all parties are spatially separated. Let us define another task \(\overline{\mathrm{CC}}_n\), where Alice-$3$ through Alice-\(n\) are located in Bob's laboratory [see Fig.\ref{fig}]. Evidently, if a function cannot be evaluated in \(\overline{\mathrm{CC}}_n\) scenario with $1$-bit communication from both Alice-$1$ and Alice-$2$ to Bob, then this function cannot be evaluated in the \(\mathrm{CC}_n\) scenario with the same amount of classical communication from both Alice-$1$ and Alice-$2$ to Bob, independent of the fact what other resources are allowed between Bob and the remaining Alices. In the \(\overline{\mathrm{CC}}_n\) scenario, consider the sub-task of evaluating \(f^\circ_n\) where \(\mathbf{x}_j = 00\) for all \(j \geq 3\). A perfect evaluation of \(f^\circ_n\) implies a perfect evaluation of \(f_2\) in the \(\mathrm{CC}_2\) task with $1$-bit communication from both Alice-$1$ and Alice-$2$ to Bob, which has already been shown to be impossible in Theorem \ref{theo2}. This concludes the proof of present theorem.
\end{proof}
 
It is natural to ask how robust is the quantum advantage established in Theorems \ref{theo2} and \ref{theo3} against the noise. Often in experiment the ideal GHZ state becomes noisy as \( G^{(p)}_{n+1} := (1-p)\ket{G_{n+1}}\bra{G_{n+1}} + p(\mathbb{I}_2/2)^{\otimes n+1}\) due to the presence of white noise, where $p\in[0,1]$. Recall that a classical strategy \(\mathcal{S}\) consists of an encoding-decoding tuple \(\mathcal{S}\equiv\{c_i=\mathcal{E}_i({\bf x}_i),\mathcal{D}_0(c_1,\cdots,c_{n}),\mathcal{D}_1(c_1,\cdots,c_{n})\}\). Let `\( b \)' denote the bit value evaluated by Bob under strategy \(\mathcal{S}\), and the success of this strategy is quantified as
\begin{align*}
P^n_{\mathcal{S}} := \sum_{\{{\bf x}_i\},{\bf y}} p({\bf x}_1,\cdots,{\bf x}_{n},{\bf y}) p(b=f_n | {\bf x}_1,\cdots,{\bf x}_{n},{\bf y}).   
\end{align*}
Given the classical resource \(\mathcal{R} \equiv \{\text{Each~Alice} \to \text{Bob} : 1~\text{bit}~\text{channel}\}\) and classical shared randomness, the parties can implement various strategies \(\{\mathcal{S}_k\}\). The optimal success with resource \(\mathcal{R}\) is then defined as \( P^n_{\mathcal{R}} := \max_{\{\mathcal{S}_k\}} P^n_{\mathcal{S}_k} \). A quantum strategy \(\mathcal{Q}\) using resources \(\{\mathcal{R} \cup \psi_{A_1\cdots B}\}\) will exhibit an advantage in the communication complexity task \(\mathrm{CC}_n\) if \( P^n_{\mathcal{Q}} > P^n_{\mathcal{R}} \). In such a case, it is evident that the nonclassicality of the shared quantum state \(\psi_{A_1 \cdots B} \in \mathcal{D}(\mathbb{C}^{d_{A_1}}\otimes \cdots \otimes \mathbb{C}^{d_B})\) underlies the quantum advantage. With this we can now prove our next result. 
\begin{theorem}\label{theo4}
For \( n \geq 2 \), the states \( G^{(p)}_{n+1} \), with \( 0 \leq p < 1/2 \), offer an advantage in evaluating the function \( f_{n} \) when supplemented with the classical resource \( \mathcal{R} \).  
\end{theorem}
\begin{proof}
We start by deriving nontrivial bounds on the optimal success probability $P^n_{\mathcal{R}}$ when the parties avail classical resource $\mathcal{R}$ only:
\begin{align*}
P^{n}_{\mathcal{R}}&\ge\frac{1}{2}+\frac{1}{\lceil 2^{(n+1)/2} \rceil}:=P^{n}_{lower},\\
P^{n}_{\mathcal{R}}&\leq \frac{3}{4}:= P^{n}_{upper},~\mbox{equality~for~} n=2,3. 
\end{align*}
To establish the lower bound $P^{n}_{lower}$ for generic $n$, it suffices to present an explicit protocol utilizing the resource $\mathcal{R}$. As demonstrated in Theorem \ref{theo1}, the function $f_n = \oplus_{i=1}^{n} x^1_i \oplus y^1 \oplus \mathbb{P}[(\sum_{i=1}^{n} x^0_i + y^0)/2]$ can be exactly evaluated with the resource $\{\mathcal{R} \cup \ket{G_{n+1}}\}$, since the condition $\oplus_{i=1}^{n} \mathcal{O}(x^0_i) \oplus \mathcal{O}(y^0) = \mathbb{P}[S_z/2]$ is perfectly satisfied by the GHZ state. Recall that, this particular condition is the requirement of $(n+1)$-player Pseudo-Telepathy game introduced by Mermin \cite{Mermin1990(1)}, where the maximum achievable success probability with classical correlation is bounded above by $P^{n}_{lower}$ (see also \cite{Brassard2003}), thereby establishing that $P^n_{\mathcal{R}} \geq P^{n}_{lower}$ for all $n \geq 2$.

To determine the upper bound, let us first consider the case $n = 2$. Only a limited number of extreme encoding-decoding strategies are possible in this case. A systematic analysis reveals that $P^{2}_{\mathcal{R}} \leq 3/4$ (see Appendix), hence $P^{2}_{upper}=3/4$. For $n\ge3$, we can apply a similar reasoning as of Theorem \ref{theo3}. Namely, instead of \(\mathrm{CC}_n\) if we consider the task \(\overline{\mathrm{CC}}_n\), then $P^{n}_{upper}$ in the barred task can be at most $3/4$; otherwise it will imply that $P^{2}_{upper}>3/4$, a contradiction. Therefore in the \(\mathrm{CC}_n\) task too we have $P^{n}_{upper}\le3/4$. Notably, for $n=3$ this bound is achievable as it matches with $P^{3}_{lower}$. However, for $n>3$ it remains to be analyzed further whether the value $3/4$ can be achieved with resource $\mathcal{R}$. 

Now, with the resource $\{\mathcal{R} \cup G^{(p)}_{n+1}\}$ and following the protocol from Theorem \ref{theo1}, the success probability is given by $P_{\mathcal{Q}} = (2-p)/2$. Considering the optimal success $P^n_{\mathcal{R}}$ of \(\mathrm{CC}_n\) to be its upper bound $3/4$, we obtain that the resource $\{\mathcal{R}\cup G^{(p)}_{n+1}\}$ is advantageous over $\mathcal{R}$, i.e., $P_{\mathcal{Q}}>3/4$ for the parameter range $p \in [0, 1/2)$. This completes the proof.            
\end{proof}
~\vspace{-.6cm}\\
The states \( G^{(p)}_{n+1} \) are known to be fully separable if and only if \( 1/[1 + 2^{-n}] \leq p \leq 1 \) \cite{Dur2000(1), Schack2000}, whereas they are genuinely multipartite entangled if and only if \( 0 \leq p < 1/[2 - 2^{-n}]\) \cite{Guhne2010}. Therefore, the advantage established in Theorem \ref{theo4} is confined to the parameter range associated with genuinely multipartite entanglement. However, the potential for advantage with non-genuine states remains open. For \( n \ge 4 \), if \( P^{(n)}_{\mathcal{R}}<3/4\), advantage may still be achieved with non-genuine but inseparable \( G^{(p)}_{n+1} \) (see Fig.\ref{figm}).

\section{DISCUSSIONS}
Quantum advantages are quite difficult to establish and even more challenging to validate experimentally. For instance, while quantum computing is anticipated to offer exponential speedups, such as in integer factorization, these expectations rest on unproven mathematical assumptions and require the development of fault-tolerant quantum computers -- an exceedingly difficult task at present \cite{Gidney2021}. In contrast, quantum resources often promise demonstrable advantages in communication complexity problems involving distributed computing tasks \cite{Brukner2004,Yamasaki2018,Tavakoli2020,Ho2022}. Within this setup, our study reveals a scalable advantage of multipartite entanglement, specifically the multi-qubit GHZ states, in reducing communication overhead while evaluating functions with inputs distributed across multiple distant servers. Remarkably, this advantage persists even under noisy conditions. Multiple research groups have already generated GHZ states across various architectures, with ongoing progress toward scaling subsystem numbers \cite{Gao2010,Zhong2018,Omran2019,Song2019,Pogorelov2021,Bao2024}. The robustness of the present protocol against noise highlights its strong potential for practical experimental implementation in these systems.

Our study opens several promising directions for future research. For instance, exploring the robustness of Theorem \ref{theo4} across generic $\mathrm{CC}_n$ tasks could yield deeper insights into the nonclassical resources underpinning the observed advantages. Additionally, the perfect advantage demonstrated in Theorem \ref{theo1} likely requires an $(n+1)$-qubit GHZ state. A potential route to prove this could involve extending the concept of self-testing \cite{Mayers2004}, where quantum states and corresponding measurements are uniquely determined (up to local isomorphism) based on Bell game statistics \cite{Supi2020}. Notably, GHZ game statistics have been shown to self-test the GHZ state \cite{Breiner2019}. Extending this framework from Bell games to scenarios involving limited communication may provide a basis for proving that the $(n+1)$-qubit GHZ state is necessary for achieving perfect advantage in $\mathrm{CC}_n$ tasks. Another compelling question is whether the 3-qubit GHZ state can further reduce communication complexity, as suggested by Theorems \ref{theo1} and \ref{theo2} -- an issue which has been addressed positively in  \cite{Chakraborty2024(1)}.

\onecolumngrid
\section{Appendix}
\subsection{Optimal success in $\mathrm{CC_2}$ with the resource $\mathcal{R}$}
\begin{center}
\begin{table}[b!]
\begin{tabular}{|c||c|c|c|c|c|c|c|c|c|c|c|c|c|c|c|c|}
\hline
\backslashbox{~$\mathcal{E}_1$~}{~$\mathcal{E}_2$~}& $\mathcal{E}^{0}_2$&~~$P_S$~~&$\mathcal{E}^{8}_2$&~~$P_S$~~& $\mathcal{E}^{4}_2$&~~$P_S$~~& $\mathcal{E}^{12}_2$&~~$P_S$~~& $\mathcal{E}^{14}_2$ &~~$P_S$~~&$\mathcal{E}^{2}_2$&~~$P_S$~~&$\mathcal{E}^{6}_2$&~~$P_S$~~& $\mathcal{E}^{10}_2$&~~$P_S$~~\\ \hline\hline
$\mathcal{E}^{0}_1$&~$\mathcal{D}^0_0,\mathcal{D}^0_1$~&$\frac{1}{2}$&~~ $\mathcal{D}^0_0,\mathcal{D}^0_1$~~&$\frac{1}{2}$&$\mathcal{D}^0_0,\mathcal{D}^0_1$&$\frac{1}{2}$&$\mathcal{D}^0_0,\mathcal{D}^0_1$&$\frac{1}{2}$&$\mathcal{D}^0_0,\mathcal{D}^0_1$&$\frac{1}{2}$&$\mathcal{D}^0_0,\mathcal{D}^0_1$&$\frac{1}{2}$&$\mathcal{D}^0_0,\mathcal{D}^0_1$&$\frac{1}{2}$&$\mathcal{D}^0_0,\mathcal{D}^0_1$&$\frac{1}{2}$\\\hline
$\mathcal{E}^{8}_1$&$\mathcal{D}^0_0,\mathcal{D}^0_1$&$\frac{1}{2}$&$\mathcal{D}^0_0,\mathcal{D}^0_1$&$\frac{1}{2}$&$\mathcal{D}^0_0,\mathcal{D}^0_1$&$\frac{1}{2}$&$\mathcal{D}^0_0,\mathcal{D}^0_1$&$\frac{1}{2}$&$\mathcal{D}^0_0,\mathcal{D}^0_1$&$\frac{1}{2}$&$\mathcal{D}^0_0,\mathcal{D}^0_1$&$\frac{1}
{2}$&$\mathcal{D}^0_0,\mathcal{D}^0_1$&$\frac{1}
{2}$&$\mathcal{D}^0_0,\mathcal{D}^0_1$&$\frac{1}{2}$\\\hline
$\mathcal{E}^{4}_1$&$\mathcal{D}^0_0,\mathcal{D}^0_1$&$\frac{1}{2}$&$\mathcal{D}^0_0,\mathcal{D}^0_1$&$\frac{1}{2}$&\cellcolor{blue!10}$\mathcal{D}^0_0,\mathcal{D}^{13}_1$&\cellcolor{blue!10}$\frac{3}{4}$ &\cellcolor{blue!10}$\mathcal{D}^{12}_0,\mathcal{D}^{0}_1$&\cellcolor{blue!10}$\frac{3}{4}$ &\cellcolor{blue!10}$\mathcal{D}^{12}_0,\mathcal{D}^{13}_1$&\cellcolor{blue!10}$\frac{3}{4}$& \cellcolor{blue!10}$\mathcal{D}^{13}_0,\mathcal{D}^{13}_1$&\cellcolor{blue!10}$\frac{3}{4}$\cellcolor{blue!10}&\cellcolor{blue!10}$\mathcal{D}^{12}_0,\mathcal{D}^{13}_1$    &\cellcolor{blue!10}$\frac{3}{4}$&\cellcolor{blue!10}$\mathcal{D}^{12}_0,\mathcal{D}^{12}_1$   &\cellcolor{blue!10}$\frac{3}{4}$\\\hline
$\mathcal{E}^{12}_1$       & $\mathcal{D}^0_0, \mathcal{D}^0_1 $ & $ \frac{1}{2}$ & $\mathcal{D}^0_0, \mathcal{D}^0_1 $ & $ \frac{1}{2}$ &\cellcolor{blue!10} $\mathcal{D}^{12}_0, \mathcal{D}^{0}_1 $  & \cellcolor{blue!10}$ \frac{3}{4}$ &\cellcolor{blue!10} $\mathcal{D}^0_0, \mathcal{D}^{12}_1 $    & \cellcolor{blue!10} $ \frac{3}{4}$ & \cellcolor{blue!10}$\mathcal{D}^{12}_0, \mathcal{D}^{12}_1 $ &\cellcolor{blue!10} $ \frac{3}{4}$ &\cellcolor{blue!10} $\mathcal{D}^{12}_0, \mathcal{D}^{13}_1 $ &\cellcolor{blue!10} $ \frac{3}{4}$ & \cellcolor{blue!10}$\mathcal{D}^{12}_0, \mathcal{D}^{12}_1 $ & \cellcolor{blue!10}$ \frac{3}{4}$ &\cellcolor{blue!10} $\mathcal{D}^{13}_0, \mathcal{D}^{12}_1 $ & \cellcolor{blue!10} $ \frac{3}{4}$ \\ \hline
$\mathcal{E}^{14}_1$         & $\mathcal{D}^0_0, \mathcal{D}^0_1 $ & $ \frac{1}{2}$ & $\mathcal{D}^0_0, \mathcal{D}^0_1 $ & $ \frac{1}{2}$ &\cellcolor{blue!10} $\mathcal{D}^{12}_0, \mathcal{D}^{13}_1 $ &\cellcolor{blue!10} $ \frac{3}{4}$ &\cellcolor{blue!10} $\mathcal{D}^{12}_0, \mathcal{D}^{12}_1 $ & \cellcolor{blue!10}$ \frac{3}{4}$ & $\mathcal{D}^{12}_0, \mathcal{D}^{0}_1 $  & $ \frac{5}{8}$ & $\mathcal{D}^0_0, \mathcal{D}^{13}_1 $    & $ \frac{5}{8}$ & $\mathcal{D}^{12}_0, \mathcal{D}^{0}_1 $      & $ \frac{5}{8}$ & $\mathcal{D}^0_0, \mathcal{D}^{12}_1 $         & $ \frac{5}{8}$ \\ \hline
$\mathcal{E}^{2}_1$      & $\mathcal{D}^0_0, \mathcal{D}^0_1 $ & $ \frac{1}{2}$ & $\mathcal{D}^0_0, \mathcal{D}^0_1 $ & $ \frac{1}{2}$ &\cellcolor{blue!10} $\mathcal{D}^{13}_0, \mathcal{D}^{13}_1 $ &\cellcolor{blue!10} $ \frac{3}{4}$ &\cellcolor{blue!10} $\mathcal{D}^{12}_0, \mathcal{D}^{13}_1 $ &\cellcolor{blue!10} $ \frac{3}{4}$ & $\mathcal{D}^0_0, \mathcal{D}^{13}_1 $    & $ \frac{5}{8}$ & $\mathcal{D}^{13}_0, \mathcal{D}^{0}_1 $  & $ \frac{5}{8}$ & $\mathcal{D}^0_0, \mathcal{D}^{13}_1 $        & $ \frac{5}{8}$ & $\mathcal{D}^{12}_0, \mathcal{D}^{0}_1 $       & $ \frac{5}{8}$ \\ \hline
$\mathcal{E}^{6}_1$  & $\mathcal{D}^0_0, \mathcal{D}^0_1 $ & $ \frac{1}{2}$ & $\mathcal{D}^0_0, \mathcal{D}^0_1 $ & $ \frac{1}{2}$ & \cellcolor{blue!10}$\mathcal{D}^{12}_0, \mathcal{D}^{13}_1 $ &\cellcolor{blue!10} $ \frac{3}{4}$ & \cellcolor{blue!10}$\mathcal{D}^{12}_0, \mathcal{D}^{12}_1 $ &\cellcolor{blue!10} $ \frac{3}{4}$ & $\mathcal{D}^{12}_0, \mathcal{D}^{0}_1 $  & $ \frac{5}{8}$ & $\mathcal{D}^0_0, \mathcal{D}^{13}_1 $    & $ \frac{5}{8}$ & $\mathcal{D}^{12}_0, \mathcal{D}^{0}_1 $      & $ \frac{5}{8}$ & $\mathcal{D}^0_0, \mathcal{D}^{12}_1 $         & $ \frac{5}{8}$ \\ \hline
$\mathcal{E}^{10}_1$ & $\mathcal{D}^0_0, \mathcal{D}^0_1 $ & $ \frac{1}{2}$ & $\mathcal{D}^0_0, \mathcal{D}^0_1 $ & $ \frac{1}{2}$ &\cellcolor{blue!10} $\mathcal{D}^{12}_0, \mathcal{D}^{12}_1 $ &\cellcolor{blue!10} $ \frac{3}{4}$ & \cellcolor{blue!10}$\mathcal{D}^{13}_0, \mathcal{D}^{12}_1 $ & \cellcolor{blue!10}$ \frac{3}{4}$ & $\mathcal{D}^0_0, \mathcal{D}^{12}_1 $    & $ \frac{5}{8}$ & $\mathcal{D}^{12}_0, \mathcal{D}^{0}_1 $  & $ \frac{5}{8}$ & $\mathcal{D}^0_0, \mathcal{D}^{12}_1 $        & $ \frac{5}{8}$ & $\mathcal{D}^{13}_0, \mathcal{D}^{0}_1 $       & $ \frac{5}{8}$ \\ \hline
\end{tabular}
\caption{The optimal success of the $\mathrm{CC_2}$ task using the resource $\mathcal{R}$. For each encoding pair $(\mathcal{E}^p_1, \mathcal{E}^q_2) \in \mathcal{G}^{{\bf x}_1}_{E} \times \mathcal{G}^{{\bf x}_2}_{E}$, we provide a corresponding decoding pair and the associated optimal success probability. Strategies achieving the maximum success rate of $3/4$ are highlighted. For all other encoding strategies, the success probability is strictly lower than this optimal value.}\label{tab1}
\end{table}
\end{center}
As already mentioned, a function $g:\{0,1\}^{\times 2} \mapsto \{0,1\}$ can be expressed as 
\begin{align}
g^{\alpha,\beta,\gamma,\delta}(u,v):&=(\alpha\land u) \oplus (\beta\land v) \oplus (\gamma\land u\land v) \oplus \delta\nonumber\\
&=\alpha u \oplus \beta v \oplus \gamma u v \oplus \delta,\label{ap1}    
\end{align}    
with $\alpha, \beta, \gamma, \delta, u, v \in \{0,1\}$. We can denote $g^{\alpha,\beta,\gamma,\delta}$ as $g^m$, where $m:= (2)^3\alpha + (2)^2\beta + (2)^1\gamma + (2)^0\delta$~. Accordingly we have $\mathcal{G}^{uv}:=\mathcal{G}^{uv}_{E}\cup\mathcal{G}^{uv}_{O}$, where
\begin{subequations}
\begin{align}
\mathcal{G}^{uv}_{E}&:=\left\{\!\begin{aligned}
g^0(u,v)&=0;~~~~~~~~~g^2(u,v)=u\land v;\\
g^4(u,v)&=v;~~~~~~~~~g^6(u,v)=\overline{u}\land v;\\
g^8(u,v)&=u;~~~~~~~~g^{10}(u,v)=u\land\overline{v};\\
g^{12}(u,v)&=u\oplus v;~~g^{14}(u,v)=u\lor v
\end{aligned}\right\},\\
\mathcal{G}^{uv}_{O}&:=\left\{
g^{2k+1}(u,v)=\overline{g^{2k}(u,v)}|~k\in\{0,\cdots,7\}\right\}.
\end{align}
\end{subequations}
While evaluating the function $f_2({\bf x}_1,{\bf x}_2,{\bf y})$ in $\mathrm{CC_2}$ task by Bob with the help of $1$ bit communication each from Alice-1 and Alice-2, they can follow a deterministic strategy as specified by a four tuple 
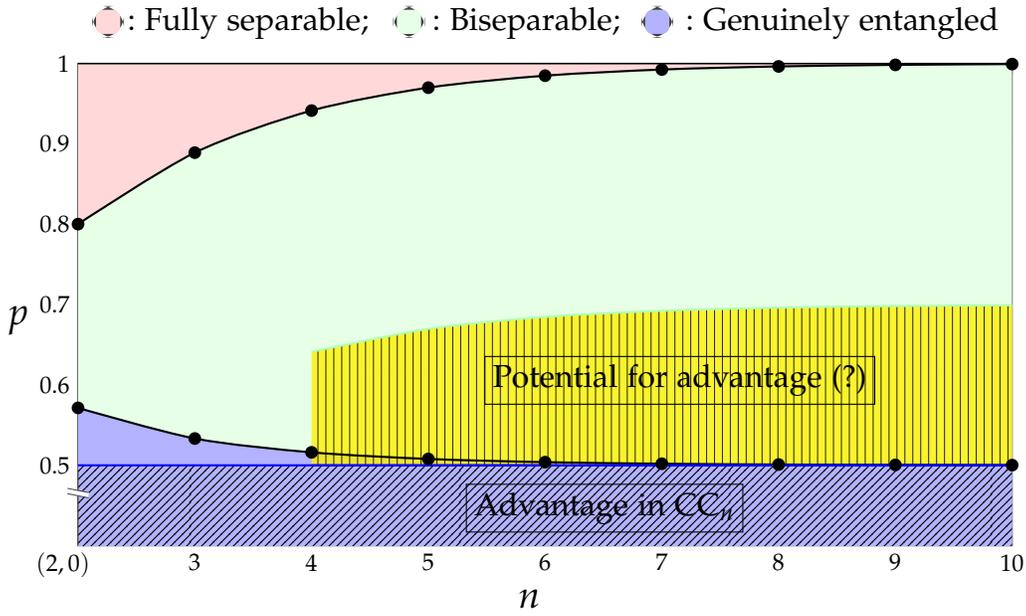
\begin{figure}[t!]
\centering
\begin{tikzpicture}
\begin{axis}[width=14cm,height=8cm,xmin=2,xmax=10,ymin=.4,ymax=1,
xtick={3,4,5,6,7,8,9,10},
ytick={.5,.6,.7,.8,.9,1.0}]
\plot[name path=A,smooth,thick,black] plot coordinates {(2,1)(10,1)};
\plot[name path=B,smooth,thick,mark=*,black] plot coordinates {(2,4/5)(3,8/9)(4,16/17)(5,32/33)(6,64/65)(7,128/129)(8,256/257)(9,512/513)(10,1024/1025)};
\plot[name path=C,smooth,thick,mark=*,color=black] plot coordinates {(2,4/7)(3,8/15)(4,16/31)(5,32/63)(6,64/127)(7,128/255)(8,256/511)(9,512/1023)(10,1024/2047)};
\plot[name path=D,smooth,thick,blue] plot coordinates {(2,.5)(10,.5)};
\plot[name path=D1,smooth,thick,blue] plot coordinates {(4,.5)(10,.5)};
\plot[name path=E,smooth,thick,mark=,color=green!30] plot coordinates {(4,16/17-.3)(5,32/33-.3)(6,64/65-.3)(7,128/129-.3)(8,256/257-.3)(9,512/513-.3)(10,1024/1025-.3)};
\plot[name path=G,smooth,thick,black] plot coordinates {(2,0)(10,0)};

\plot[fill=red!15] fill between[of=A and B];
\plot[fill=green!10] fill between[of=B and C];
\plot[fill=blue!30] fill between[of=C and G];
\plot[black,fill=blue!30,postaction=
{pattern=north east lines}] fill between[of=D and G];
\plot[black,fill=yellow!100,opacity=0.8,postaction=
{pattern=vertical lines}] fill between[of=E and D1];
\end{axis}
\draw[thin,black!70,TwoMarks=0.70] (0,0) -- (0,1);

\node[] at (.4, 6.97){\begin{tcolorbox}[width=12pt,height=12pt,colback=red!15,outer arc=3mm]\end{tcolorbox}};
\node[] at (2.2, 6.95){\large{~:~Fully separable;}};

\node[] at (4.4, 6.97){\begin{tcolorbox}[width=12pt,height=12pt,colback=green!10,outer arc=3mm]\end{tcolorbox}};
\node[] at (5.9, 6.95){\large{~:~Biseparable;}};

\node[] at (7.7, 6.97){\begin{tcolorbox}[width=12pt,height=12pt,colback=blue!30,outer arc=3mm]\end{tcolorbox}};
\node[] at (10.1, 6.95){\large{~:~Genuinely entangled}};

\node[draw] at (7, .5){\large{Advantage in $\mathrm{CC}_n$}};
\node[draw] at (8, 2.2)  {\large{Potential for advantage (?)}};
\node[] at (6, -0.7){\Large\bf{{$n$}}};
\node[] at (-.8,3){\Large\bf{{$p$}}};
\node[] at (-.2,-.25){$(2,0)$};

\end{tikzpicture}
\caption{ {\bf Noise tolerance of entanglement-based protocol:} We consider the white noise model for noisy GHZ states, defined as \( G^{(p)}_{n+1} := (1-p)\ket{G_{n+1}}\bra{G_{n+1}} + p(\mathbb{I}_2/2)^{\otimes n+1} \). These states are fully separable if and only if \( 1/[1 + 2^{-n}] \leq p \leq 1 \), and they remain genuinely multipartite entangled if and only if \( 0 \leq p < 1/[2 - 2^{-n}] \) \cite{Dur2000(1), Schack2000,Guhne2010}. For any $n>2$, the states provide an advantage in the corresponding $\mathrm{CC_n}$ task for $p \in [0,1/2]$. For $n=2$ and $n=3$, this advantage is not possible for $p > 1/2$. However, for larger $n$, such possibilities are not excluded, and the quantum protocol may exhibit greater noise robustness. Furthermore, robustness analysis of the quantum protocol can also be extended to colored noise models.}\label{figm}
\end{figure} 
\begin{align}
&\mathcal{S}(p,q,r,s)\equiv(\mathcal{E}^p_1,\mathcal{E}^q_2,\mathcal{D}^r_0,\mathcal{D}^s_1)\in\mathcal{G}^{{\bf x}_1}\times\mathcal{G}^{{\bf x}_2}\times\mathcal{G}^{\bf c}\times\mathcal{G}^{\bf c},\nonumber\\
&{\bf c}:=c_1c_2,~~c_1:=\mathcal{E}^p_1(x^0_1,x^1_1)~\&~c_2:=\mathcal{E}^q_2(x^0_2,x^1_2).
\end{align}
Note that the function $g^{\alpha,\beta,\gamma,\delta}$ in Eq.(\ref{ap1}) have the following property:
\begin{subequations}\label{app}
\begin{align}
g^{\alpha,\beta,\gamma,\delta}(u,v)&=g^{\alpha\oplus\gamma,\beta,\gamma,\delta\oplus\beta}(u,\overline{v})\\
&=g^{\alpha,\beta\oplus\gamma,\gamma,\delta\oplus\alpha}(\overline{u},v)\\
&=g^{\alpha\oplus\gamma,\beta\oplus\gamma,\gamma,\delta\oplus\alpha\oplus\beta\oplus\gamma}(\overline{u},\overline{v}).
\end{align} 
\end{subequations}
Due to property (\ref{app}), without loss of any generality, we can restrict Alices' encoding $(\mathcal{E}^p_1,\mathcal{E}^q_2)$ in $\mathcal{G}^{{\bf x}_1}_{E}\times \mathcal{G}^{{\bf x}_2}_{E}$, while decodings on Bob's end being arbitrary. Therefore, we can restrict our analysis  only for the strategies
\begin{align}
\mathcal{S}(p,q,r,s)\in\mathcal{G}_E^{{\bf x}_1}\times\mathcal{G}_E^{{\bf x}_2}\times\mathcal{G}^{\bf c}\times\mathcal{G}^{\bf c}.
\end{align}
Given a fixed encoding pair $(\mathcal{E}^{p^\star}_1,\mathcal{E}^{q^\star}_2)$ for Alices, Bob needs to choose suitable decodings $(\mathcal{D}^{r^\star}_0,\mathcal{D}^{s^\star}_1)$ that optimize the success probability $P^2_{\mathcal{S}^\star}$ of evaluating the function $f_2({\bf x}_1,{\bf x}_2,{\bf y})$.

{\it Case study:} Consider that both Alice-1 and Alice-2 transmit the second bit of their respective input strings, i.e., $c_1 = \mathcal{E}^4_1$ and $c_2 = \mathcal{E}^4_2$. The promise condition $x^0_1 \oplus x^0_2 \oplus y^0 = 0$ imposes restrictions on inputs and outputs:
\begin{subequations}
\begin{align}
y^0 = 0~:~\left\{\!\begin{aligned} &~~~~~~~~~~~~~~x^0_1x^0_2 = \{00,11\}, \\
&f_2(0x^1_1, 0x^1_2, 0)= x^1_1 \oplus x^1_2,\\ 
&f_2(1, x^1_1, 1x^1_2, 0)= x^1_1 \oplus x^1_2 \oplus 1\end{aligned}\right\},\\
y^0 = 1~:~\left\{\!\begin{aligned} &~~~~~~~~~~~~~~x^0_1x^0_2 = \{01,10\}, \\
&f_2(0x^1_1, 1x^1_2, 1)= x^1_1 \oplus x^1_2 \oplus 1,\\ 
&f_2(1x^1_1, 0x^1_2, 1)= x^1_1 \oplus x^1_2 \oplus 1\end{aligned}\right\}.       \end{align} 
\end{subequations}
For $y^0 = 1$, Bob can perfectly evaluate the function using communication from the Alices and the knowledge of promise condition by selecting the decoding strategy $\mathcal{D}^{13}_1$. However, for $y^0 = 0$, with the functions $x^1_1 \oplus x^1_2$ and $x^1_1 \oplus x^1_2 \oplus 1$ being balanced and complementary, Bob can only evaluate the function with a success probability at most $1/2$. In this case, Bob may choose different decoding, such as $\mathcal{D}^0_0, \mathcal{D}^1_0, \mathcal{D}^{12}_0$, or $\mathcal{D}^{13}_0$. Therefore, the strategy $\mathcal{S}(4,4,0/1/12/13,13)$ yields a maximum average success probability $P_{\mathcal{S}} = 3/4$.

A similar analysis can be applied for any encoding pair $(\mathcal{E}^p_1, \mathcal{E}^q_2) \in \mathcal{G}^{{\bf x}_1}_{E} \times \mathcal{G}^{{\bf x}_2}_{E}$, and the optimal success probability can be efficiently determined through computational methods. As shown in Table \ref{tab1}, the success probability is upper bounded by $3/4$ in each case.

\twocolumngrid
~\vspace{0cm}\\
{\bf Acknowledgment}: KA acknowledges support from the CSIR project {$09/0575(19300)/2024$-EMR-I. SGN acknowledges support from the CSIR project $09/0575(15951)/2022$-EMR-I. MB acknowledges funding from the National Mission in Interdisciplinary Cyber-Physical systems from the Department of Science and Technology through the I-HUB Quantum Technology Foundation (Grant no: I-HUB/PDF/2021-22/008).


%

\end{document}